\documentclass[a4paper]{article}

\pdfoutput=1

\usepackage{microtype}

\usepackage{latexsym,color,graphicx,amssymb,fullpage}
\usepackage[fleqn]{amsmath}
\usepackage{algpseudocode}
\usepackage{amsfonts}
\usepackage{algorithm}
\usepackage{comment}
\usepackage{caption}
\usepackage{amsthm}
\usepackage{subfig}
\usepackage{enumitem}
\DeclareCaptionFormat{algor}{%
  \hrulefill\par\offinterlineskip\vskip1pt%
    \textbf{#1#2}#3\offinterlineskip\hrulefill}
\DeclareCaptionStyle{algori}{singlelinecheck=off,format=algor,labelsep=space}
\newtheorem{theorem}{Theorem}[section]
\newtheorem{lemma}[theorem]{Lemma}
\newtheorem{corollary}[theorem]{Corollary}
\newtheorem{remark}[theorem]{Remark}
\newtheorem{definition}[theorem]{Definition}

\setlist[enumerate]{noitemsep}
\setlist[itemize]{noitemsep}

\def\andrew#1{\marginpar{{\bf AG:} #1}}

\def\hl{\text{\sc HL}}
\def\hhl{\text{\sc HHL}}
\def\opt{\text{\sc OPT}}

\newcommand{\pset}{{\cal P}}

\newcommand{\setsystem}{{\cal S}}
\newcommand{\rs}{$r$-significant}
\newcommand{\shs}{SPHS}

\newcommand{\Xcomment}[1]{}
\newcommand{\dist}{\ensuremath{\text{dist}}}

\usepackage{etex,etoolbox}
\usepackage{amsthm,amssymb}
\usepackage{thmtools}
\usepackage{environ}
\makeatletter
\newcommand\fixstatement[2]{\AtEndEnvironment{#1}{\xdef\pat@label{#2\space\@currentlabel}}}
\globtoksblk\prooftoks{1000}
\newcounter{proofcount}
\newenvironment{proofatend}{\begin{proof}}{\end{proof}}
\NewEnviron{not@proofatend}{%
  \edef\next{%
    \noexpand\begin{proof}[Proof of \pat@label]%
    \unexpanded\expandafter{\BODY}}%
  \global\toks\numexpr\prooftoks+\value{proofcount}\relax=\expandafter{\next\end{proof}}
  \stepcounter{proofcount}}
\def\printproofs{%
  \count@=\z@
  \loop
    \the\toks\numexpr\prooftoks+\count@\relax
    \ifnum\count@<\value{proofcount}%
    \advance\count@\@ne
  \repeat}
\makeatother
\fixstatement{theorem}{Theorem}
\fixstatement{lemma}{Lemma}
\fixstatement{corollary}{Corollary}

\begin{document}

\title{On the Complexity of Hub Labeling}

\author{%
\makebox[.25\linewidth]{Maxim Babenko\thanks{Yandex and Higer School of Economics, Moscow, Russia. E-mail: \texttt{maxim.babenko@gmail.com}}}\and%
\makebox[.25\linewidth]{Andrew V. Goldberg\thanks{Amazon.com, Inc., Palo Alto, USA. Part of the work done while the author was at Microsoft Research. E-mail: \texttt{avg@alum.mit.edu}}}\and%
\makebox[.25\linewidth]{Haim Kaplan\thanks{School of Computer Science, Tel Aviv University, Israel. E-mail: \texttt{haimk@post.tau.ac.il}}}\and%
\makebox[.25\linewidth]{Ruslan Savchenko\thanks{Yandex, Moscow, Russia. Part of the work done while the author was at Department of Mech.\ and Math., Moscow State University. E-mail: \texttt{ruslan.savchenko@gmail.com}}}\and%
\makebox[.25\linewidth]{Mathias Weller\thanks{LIRMM, Universit\'e Montpellier II, France. E-mail: \texttt{mathias.weller@lirmm.fr}}}%
}

\date{}

\maketitle

\begin{abstract}
Hub Labeling (\hl) is a data structure for distance oracles.
Hierarchical \hl\ (\hhl) is a special type of \hl, that received a lot of attention from a practical point of view.
However, theoretical questions such as NP-hardness and approximation guarantee for \hhl\ algorithms have been left aside.
In this paper we study \hl\ and \hhl\ from the complexity theory point of view. We prove that both \hl\ and \hhl\ are NP-hard, and present upper and lower bounds for the approximation ratios of greedy \hhl\ algorithms used in practice.
We also introduce a new variant of the greedy \hhl\ algorithm and a proof that it produces small labels for graphs with small highway dimension.
\end{abstract}

\section{Introduction}
The point-to-point shortest path problem is a classical optimization 
problem with many applications.
The input to the problem is a graph
$G=(V,E)$, a length function $\ell: E \rightarrow \mathbb{R}$, and
a pair $s,t \in V$.
We define $n=|V|$ and $m=|E|$.
The goal is to find $\dist(s,t)$, the length of the shortest
$s$--$t$ path in $G$, where the length of a path is the sum of the lengths of its arcs.
We assume that the length function is non-negative and that there are no zero-length cycles.

The hub labeling algorithm (\hl)~\cite{CHKZ-03,GPPR-04}
is a shortest path algorithm that computes vertex labels during
preprocessing stage and answers $s,t$ queries using only the labels
of $s$ and $t$; the input graph is not used for queries~\cite{pel-00}.
For a directed graph a \emph{label} $L(v)$ for a vertex $v \in V$
consists of the \emph{forward label} $L_f(v)$ and the
\emph{backward label} $L_b(v)$. The forward label $L_f(v)$ consists of a sequence of pairs $(w,\dist(v,w))$,
where $\dist(v,w)$ is the distance (in $G$)
from $v$ to $w$.
The backward label $L_b$ is similar, with pairs $(u,\dist(u,v))$.
Vertices $w$ and $u$ (for forward and backward labels, respectively)
are called the \emph{hubs} of $v$.
For an undirected graph $L_f = L_b$, and we denote the labeling by $L$, so $L(v)$ itself is a set of pairs $(w,\dist(v,w))$.

The labels must obey the \emph{cover property}: for any two
vertices $s$ and $t$, the set $L_f(s) \cap L_b(t)$ must contain at least
one hub $v$ that is on a shortest $s$--$t$ path (we say that $v$ \emph{covers} the $[s,t]$ pair).
Given the labels, \hl\ queries are straightforward: to find $\dist(s,t)$,
simply find the hub $v \in L_f(s) \cap L_b(t)$ that minimizes
$\dist(s,v) + \dist(v,t)$.

Query time and space complexity depends on the label size.
The size of a label $|L(v)|$ is the number of hubs it contains.
For a directed graph the size of a forward (backward) label, $|L_f(v)|$ ($|L_b(v)|$), is the
number of hubs it contains and
the size of the full label of $v$, $L(v) = (L_f(v), L_b(v))$, is $|L(v)| = |L_f(v)|+|L_b(v)|$.
Unless mentioned otherwise, preprocessing algorithms attempt to minimize
the \emph{total labeling size} $|L| = \sum_V |L(v)|$.

Cohen et al.~\cite{CHKZ-03} give an $O(\log n)$ approximation algorithm
for \hl\ preprocessing.
This algorithm was generalized in~\cite{BGGN-13-2} and
sped up in~\cite{DGSW-14}.
These approximation algorithms compute small labels but,
although polynomial, do not scale to large problems~\cite{DGSW-14}.

A special case of \hl\ is \emph{hierarchical hub labeling}
(\hhl)~\cite{ADGW-12},
where vertices are globally ranked by ``importance'' and the label
for a vertex $v$ can only have more important hubs than $v$ and $v$ itself.
\hhl\ implementations are faster in practice than general \hl\ ones.
For several important graph classes, such as road and complex networks,
\hhl\ implementations find small labelings and scale to
large problems~\cite{ADGW-11,ADGW-12,AIY-13,DGPW-14}.
However, for the algorithms used in practice such as
hierarchical greedy (g-HHL) and hierarchical weighted greedy (w-HHL)
there was no theoretical guarantee on the approximation ratio.

Most of the work on the computational complexity of \hl\ (and \hhl)
algorithms is experimental.
The exceptions are approximation algorithms for \hl\ mentioned above,
and upper bounds for
\hl\ in case of low highway dimension~\cite{AFGW-10,ADFGW-11,ADFGW-13}.
However, there was no NP-completeness proof of for \hl.
NP-completeness was implicitly conjectured in~\cite{CHKZ-03}: this assumption motivates the $O(\log n)$-approximation algorithm. In addition, in~\cite{CHKZ-03} the authors prove that a more general problem, in which the paths to cover are part of the input, is NP-complete (which does not imply NP-hardness of the original problem). 

In this paper we obtain the following results on \hl\ and \hhl\ complexity:
\begin{itemize}
\item
We show that both the optimal \hl\ and the optimal \hhl\ problems
are NP-complete.
\item
We show that in a network of highway dimension $h$ and diameter $D$,
there is an \hhl\ such that every label size is $O(h \log D)$,
matching the \hl\ bound of~\cite{ADFGW-11,ADGW-11,ADFGW-13}.
\item
We propose a variant of the greedy algorithm (called d-HHL), for which we prove
    \begin{itemize}
    \item
    an $O(h \log n \log D)$ bound for every label size,
    \item
    an $O(\sqrt n \log n \log D)$-approximation ratio compared to the
optimal \hl\ (and therefore the optimal \hhl),
    \item
    an $\Omega(\sqrt n)$ lower bound on the approximation ratio for
the optimal \hhl.
    \end{itemize}
\item
For g-HHL, we prove
    \begin{itemize}
    \item
    an $O(\sqrt n \log n)$-approximation ratio compared to the
    optimal \hl.
    \item
    an $\Omega(\sqrt n)$ lower bound on the approximation ratio for
    the optimal \hhl.
    \end{itemize}
\item
For w-HHL, we prove
    \begin{itemize}
    \item
    an $O(\sqrt n \log n)$-approximation ratio compared to the
    optimal \hl.
    \item
    $\Omega(\sqrt[3] n)$ lower bound on the approximation ratio for
    the optimal \hhl.
    \end{itemize}
\item
We give an example showing that hierarchical labelings can be
$\Omega(\sqrt n)$ bigger than general labelings, improving and
simplifying~\cite{GRS-13}.
\end{itemize}

Our lower bounds on the greedy algorithms show that they do not give a poly-log
approximation, leaving the question of the possibility of poly-log approximation
open.
This is an interesting theoretical problem that may have a practical impact
as well.

\section{Preliminaries}


\subsection{HL Approximation Algorithm}
\label{sec:approx}

Cohen et al.\ obtain their $O(\log n)$ approximation algorithm for \hl\ by formulating it as a weighted set cover problem and applying the well known greedy approximation algorithm for set cover.
In the weighted set cover problem there is a \emph{universe} set $U$,
a family $\mathcal{F}$ of some subsets of $U$, a cost function $c: \mathcal{F} \rightarrow \mathbb{R_+}$,
and the goal is to find a collection $\mathcal{C}\subseteq \mathcal{F}$
such that $\cup_{S\in \mathcal{C}} S = U$ and $\sum_{S\in \mathcal{C}} c(S)$ is minimized.
The greedy set cover algorithm starts with an empty $\mathcal{C}$, then iteratively picks a set $S$ which maximizes
the ratio of the number of newly covered elements in $U$ to the cost of $S$ and adds $S$ to $\mathcal{C}$.

The elements to cover in the equivalent set cover instance are vertex pairs $[u,v]$.
For a directed graph pairs in $U$ are ordered and for an undirected graph pairs are unordered.
We first discuss directed graphs, then undirected ones.
Every possible set $P$ of vertex pairs  such that there exists a vertex $u$ which hits a shortest path between every pair in $P$ is a set. (There are exponentially many sets, but they are not used explicitly.)
The cost of a set $P$ is the number of vertices that appear in the first component of a pair in $P$ plus
the number of vertices that appear in the second component of a pair in $P$.

The greedy approximation algorithm for set cover as applied to this set cover instance is as follows.
The algorithm  maintains the set $U$
of \emph{uncovered} vertex pairs: $[u,w] \in U$ if $L_f(u) \cap L_b(w)$
does not contain a vertex on a shortest $u$--$w$ path.
Initially $U$ contains all vertex pairs $[u,w]$ such that
$w$ is reachable from $u$. The algorithm terminates
when $U$ becomes empty. Starting with an empty labeling,
in each iteration, the algorithm adds a vertex $v$ to forward labels
of vertices in a set $ S' \subseteq V$ and to backward labels of the vertices in $ S'' \subseteq V$
such that the ratio of the number of newly-covered pairs over the total
increase in the size of the labeling is (approximately) maximized.
Formally, let $U(v, S',S'')$ be the set of pairs in $U$ which are covered
if $v$ is added to $L_f(u): u \in S'$ and $L_b(w): w \in S''$.
The algorithm maximizes $|U(v, S',S'')|/(|S'| + |S''|)$ over all $v \in V$
and $S', S'' \subseteq V$.

To find the triples $(v, S',S'')$ efficiently the algorithm uses center graphs defined as follows.
A \emph{center graph} of $v$, $G_v = (X, Y, A_v)$, is a bipartite graph
with $X=V$, $Y=V$, and an arc $(u,w) \in A_v$ if $[u,w] \in U$ and some shortest path from
$u$ to $w$ goes through $v$.
The algorithm finds $(v, S',S'')$  that maximizes $|U(v, S',S'')|/(|S'| + |S''|)$ by computing a densest subgraph among all the subgraphs of the center graphs $G_v$.
The \emph{density} of a graph $G=(V,A)$ is $|A|/|V|$.
The \emph{maximum density subgraph (MDS)} problem
is the problem of finding an (induced) subgraph of a given graph $G$ of maximum density. This problem
can be solved in
polynomial time using parametric flows (e.g.,~\cite{GGT-89}).
For a vertex $v$, the
arcs of a subgraph of $G_v$ induced by $S' \subseteq X$ and $S'' \subseteq Y$
correspond to the pairs of vertices in $U$ that become covered
if $v$ is added to $L_f(u): u \in S'$ and $L_b(w): w \in S''$.
Therefore, the MDS of $G_v$ maximizes $|U(v, S',S'')|/(|S'| + |S''|)$
over all $S'$, $S''$.

For undirected graphs we have $L_f(v)=L_b(v)=L(v)$ by definition.
Pairs $[u,v]\in U$ are unordered and the cost of a set $P$ of unordered vertex pairs is the number of vertices that appear in a pair in $P$.
Let $U(v,S)$ be the set of unordered vertex pairs that become covered if we add $v$ to $L(u): u\in S$.
We want to maximize $U(v,S)/|S|$.
To find such a tuple, we use another type of a
center graph of $v$, $G_v=(V,E_v)$.
$G_v$ is an undirected graph with vertex set $V$ and 
with an edge $\{u,w\}\in E_v$ if $[u,w]\in U$ and some shortest path between $u$ and $w$ goes trough $v$. 
(For a pair $[v,v]$ there is a self-loop $\{v,v\}$ in $E_v$.)
Note that $G_v$ is not necessarily bipartite.
As in the directed case, MDS of $G_v$ maximizes $U(v,S)/|S|$ over all $S$.

The following is a folklore lemma about the greedy set cover algorithm.
\begin{lemma}\label{lm:appr}
If we run the greedy set cover algorithm where in each iteration we pick a set whose coverage to cost ratio is at least $1/f(n)$ fraction of the maximum coverage to cost ratio, then we get a cover of
cost within an $O(f(n) \log n)$ factor of optimal.
\end{lemma}

Cohen et al.~\cite{CHKZ-03}
used this lemma and instead of finding the MDS exactly they used a linear-time 2-approximation algorithm~\cite{KP-94}. 
The result is an $O(\log n)$-approximation algorithm running in $O(n^5)$ time.
Delling et al.~\cite{DGSW-14} improve the running time to $O(n^3\log n)$.

\subsection{Canonical HHL}

Vertices are \emph{ordered} if there is a bijection $\pi: V\to \{1,\ldots, |V|\}$.
We say that $u$ is \emph{more important} than $v$ if $\pi(u) < \pi(v)$.
The labeling $L$ is \emph{hierarchical} if there is an order $\pi$ such that $u\in L_f(v) \cup L_b(v)$ implies $\pi(u) \le \pi(v)$.
In this case we say that $L$ \emph{respects} $\pi$.

Let $P_{u,v}$ denote the set of all vertices on shortest paths from $u$ to $v$.
For an order $\pi$ we define a \emph{canonical} \hhl\ in the following way: $u\in L_f(v)$ (resp. $u\in L_b(v)$) if and only if $u$ is the most important
vertex in $P_{v,u}$ (resp.\ $P_{u,v}$).
The following theorem is implicit in \cite{ADGW-12,DGPW-14,GRS-13}.
\begin{theorem}
For an order $\pi$ the canonical \hhl\ is the minimum \hhl\ that respects $\pi$.
\end{theorem}
\begin{proof}
We first show that the canonical \hhl\ $L$ obeys the cover property.
For a pair $[v,w]$ let $u$ be the most important vertex in $P_{v,w}$.
Consider any $v$-$u$ shortest path. It is easy to see that it is a subpath of some $v$-$w$ shortest path.
Therefore by the definition of canonical labeling we have $u\in L_f(v)$ and $u\in L_b(w)$.

Now we show that $L$ is a sublabeling of any \hhl\ $\tilde L$ that respects $\pi$.
Let $u\in L_f(v)$ (resp. $u\in L_b(v)$).
Then $u$ is more important than any other vertex $w$ on a $v$-$u$ (resp. $u$-$v$) shortest path.
Therefore $\tilde L_b(u)$ (resp.\ $\tilde L_f(u)$) doesn't have any such $w$ except $u$.
Since $\tilde L$ covers the $[v,u]$ (resp. $[u,v]$) pair we have $u\in \tilde L_f(v)$ (resp. $u\in \tilde L_b(v)$).
So $L$ is a sublabeling of $\tilde L$.
\end{proof}

\subsection{Greedy HHL Algorithms}\label{sec:greed}
In this section we describe greedy \hhl\ algorithms in terms of center graphs.
For an alternative description and efficient implementation of these algorithms,
see~\cite{ADGW-12,DGPW-14}.

A greedy \hhl\ algorithm maintains the center graphs $G_v = (X, Y, A_v)$ 
defined in Section~\ref{sec:approx}.
At each iteration, the algorithm selects a center graph of a vertex $v$
and adds $v$ to $L_f(u)$ for all non-isolated vertices $u\in X$  
and to $L_b(w)$ for all non-isolated vertices $w\in Y$.
Note that after the labels are augmented this way, 
all vertex pairs $[u,w]$ for which
there is a $u$-$w$ shortest path passing through $v$ are covered.
Therefore, the center graph of every vertex is chosen once, and the labeling
is hierarchical.

Greedy algorithms differ by the criteria used to select the next
center graph to process.
The \emph{greedy \hhl\ (g-HHL)} algorithm selects the center graph with
most edges.
The \emph{weighted greedy \hhl\ (w-HHL)} algorithm selects a center graph
with the highest density (the number of edges divided by the number of
non-isolated vertices).

We propose a new \emph{distance greedy \hhl\ (d-HHL)} algorithm.
To every vertex pair $[u,v]$ we assign a weight
$$
W(u,v) =
\begin{cases}
0, & \text{ if }\dist(u,v) = 0\\
n^{2\lfloor \log_2 (\dist(u,v))\rfloor}, & \text { otherwise}
\end{cases}
$$
and use $W$ to weight the corresponding edges in center graphs.
At each iteration, d-HHL selects a center graph with the largest sum of edge weights.

We define the \emph{level} of $[u,v]$ as
$\lfloor \log_2 (\dist(u,v)) \rfloor$ (if $\dist(u,v)=0$ the level of $[u,v]$ is~$-\infty$).
The definition of $W$ insures that if $[u,v]$ is the maximum level uncovered vertex pair,
$W(u,v)$ is greater than the total weight of all lower-level uncovered pairs.
Therefore d-HHL primarily maximizes the number of uncovered maximum level
pairs that become covered, and other pairs that become covered are used
essentially as tie-breakers.

We say that a vertex $w$ has {\em level} $i$ if at the iteration when $w$ is selected
by d-HHL, the maximum level of an uncovered vertex pair is $i$.
As the algorithm proceeds, the levels of vertices it selects are monotony
decreasing.

\subsection{Highway Dimension}

In this section we review the definition of highway dimension (HD) and related concepts.
As HD is defined for undirected graphs, when we talk about HD we assume that
all graphs are undirected and connected.

\begin{definition} \label{def:rw}
Given a shortest
path $P= (v_1, \ldots, v_k)$ and $r > 0$, a shortest path $P'$ is an $r$-witness for $P$ if and only if $\ell(P') > r$ and one of the following conditions holds:
\begin{enumerate}
\item
$P' = P$; or
\item
$P' = (v_0, v_1, \ldots, v_k)$; or
\item
$P' = (v_1, \ldots, v_k, v_{k+1})$; or
\item
$P' = (v_0, v_1, \ldots, v_k, v_{k+1})$.
\end {enumerate}
\end{definition}

\begin{definition} \label{def:rs} A shortest path $P$ is \rs\ if it has an $r$-witness path.
\end{definition}
\Xcomment{In other words, $P$ is \rs\ if it can be extended by at most one vertex at each end to
a shortest path of length greater than $r$. Note that, if $(v,w) \in E$ and $\ell(v,w) > r$, then the trivial paths $(v)$ and $(w)$ are \rs, since $(v,w)$ is the shortest path from $v$ to $w$.
}
Let $\pset_r$ denote the set of all \rs\ paths.
Given a vertex $v$ and a path $P$, we define the distance from $v$ to $P$
by $\dist(v,P)=\min_{w \in P} \dist(v,w)$.
\begin{definition}\label{def:close}
A shortest path $P$ is $(r,d)$-close to a vertex $v$ if $P$ is \rs\ with an $r$-witness path $P'$ such that $\dist(v,P') \leq d$.
\end{definition}
Note that if $P$ is $(r,d)$-close to $v$,
then $P$ is also $(r',d')$-close to $v$ for any $0<r'\leq r$ and $0\leq d \leq d'$.

Let the {\em $r$-neighborhood} of $v$,
denoted by $S_r(v)$, be the set of all $P \in \pset_r$ that are $(r,2r)$-close to $v$.
Given a set of paths $\cal{P}$,
we say that $H\subseteq V$ is a \emph{hitting set} for
$\cal{P}$ if every path in $\cal{P}$ contains a vertex in $H$.

\begin{definition}
\label{def:hd}
A network $(G,\ell)$ has {\em highway dimension (HD)} $h$ if $h$ is the
smallest integer such that
for any $r>0$ and
any $v \in V$, there exists a hitting set $H$ for $S_r(v)$ (that depends on $v$ and $r$) with $|H| \le h$.
\end{definition}

Given $r \ge 0$ and $v \in V$, we define the \emph{ball} of radius $r$
centered at $v$, $B_{r}(v)$, to be the set of all vertices within distance at most
$r$ from $v$.

A notion related to highway dimension is that of a {\em sparse shortest-path hitting set (\shs)}.
\begin{definition}
\label{def:spc}
For $r > 0$, an $(h,r)$-\shs\ is a hitting set $C \subseteq V$ for $\pset_r$
such that $\forall v \in V,\; |B_{2r}(v) \cap C| \le h$.
\end{definition}

Abraham et al.~\cite{AFGW-10,ADFGW-11,ADFGW-13} show:
\begin{theorem} \label{thm:sps}
If the highway dimension of  a network $(G, \ell)$ is $h$, then (1) for any $r > 0$, a minimum hitting
set for $P_r$ is an $(h, r)$-\shs\ and (2) If shortest paths are unique one can find an $(h\log h, r)$-\shs\ in polynomial time.
\end{theorem}

\section{HHL and Highway Dimension}
\label{hhl-hd}
Abraham et al.~\cite{AFGW-10,ADFGW-11,ADFGW-13} show that a network with
HD $h$ and diameter $D$ has an \hl\ with $|L(v)| = O(h \log D)$, and that
in polynomial time one can find an \hl\ with $|L(v)| = O(h \log h \log D)$.
We show similar results for \hhl.

Assume that  edge lengths are at least $1$ and let $D$ be the diameter of the network $(G,\ell)$.
A multiscale \shs\ of $(G,\ell)$ is a collection of sets $C_i$ for $0 \le i\le \lceil \log D\rceil$, where each $C_i$ is a
$(h, 2^{i-1})$-\shs. In particular, note that $C_0 = V$, since every vertex is an $(1/2)$-significant path.
For $0 \le i \le \lceil \log D \rceil$,
let $Q_i = C_i \setminus \bigcup_{j=i+1}^{\lceil \log D\rceil} C_j$.

\begin{theorem}
A network with HD $h$ and diameter $D$ has an \hhl\ with
$|L(v)| = O(h \log D)$ for all $v \in V$,
and if shortest paths are unique  one can find in polynomial time an \hhl\ with $|L(v)| = O(h \log h \log D)$.
\end{theorem}
\begin{proofatend}
Consider the ordering $r$ such that for $i<j$ each  $w\in Q_i$ is less
important than each $v\in Q_j$ (i.e.\ $r(w)<r(v)$), and vertices within each $Q_i$ are ordered
arbitrarily.
For each $v \in Q_i$, define
$$
L(v) = \{v\} \cup \{ r(w) > r(v), \; w \in C_j \cap B_{2^j}(v)\} .
$$
Consider a shortest $s$--$t$ path $P$ and let $i$ be
such that $2^{i-1}< \ell(P) \le 2^i$.
Assume, w.l.g., that $r(s) < r(t)$.
Let $s \in Q_x$ and $t \in Q_y$; we have $x \le y$.

If $y \ge i$, then $t \in B_{2^i}(s)$ so $t\in Q_y \cap B_{2^y}(s) \subseteq  C_y \cap B_{2^y}(s)$ and therefore $t \in L(s)$.
If $y < i$, then since $x\le y < i$  there must be  a vertex $w \not= s,t$ such that $w \in P \cap C_i$.
By the definition of $i$, $w \in B_{2^i}(s)$ and $w \in B_{2^i}(t)$.
Therefore $w \in L(s) \cap L(t)$.
In both cases, the cover property holds.

Using the multiscale \shs\ provided by Theorem \ref{thm:sps} we get that there exists
an HHL such that
 $|L(v)| = O(h \log D)$ and if shortest paths are unique we can compute in polynomial time an HHL such that
 $|L(v)| = O(h \log h \log D)$.
\end{proofatend}

Next we discuss the \emph{distance greedy} d-HHL algorithm 
(defined in Section~\ref{sec:greed}).

\begin{theorem}\label{tm:dhhl}
In a network with HD $h$ and diameter $D$
d-HHL finds a labeling with $|L(v)| = O(h \log n \log D)$, for all $v \in V$.
\end{theorem}
\begin{proofatend}
We show that for every vertex  $v$ and level $i$, $L(v)$ contains $O(h \log n)$
hubs at level $i$.

Consider the (consecutive) iterations of the algorithm that select
vertices at level $i$.
Consider $v \in V$ and $B_{2\cdot 2^i}(v)$.
Since d-HHL already covered all vertex pairs of level greater than $i$, $v$ can accumulate hubs of level $i$ only from
vertices in $B_{2\cdot 2^i}(v)$.

Suppose at some step the algorithm chooses a level $i$ vertex $w$ in $B_{2\cdot 2^i}(v)$.
Every $x$--$y$ shortest path of length $\ge 2^i$ hit by $w$ is in $\setsystem_{2^i}(v)$.
By the definition of highway dimension, there is a hitting set $H$
for $\setsystem_{2^i}(v)$ with $|H| \le h$.

We call a yet uncovered vertex pair $[x,y]$  {\em relevant} if there is a $x$--$y$ shortest path in $\setsystem_{2^i}(v)$ and
$\dist(x,y) \ge 2^i$.
Since $H$ is a hitting set for $\setsystem_{2^i}(v)$, $H$ is also a hitting set for the set of relevant vertex pairs (it hits a shortest path between each such pair).
It follows that there is a vertex $u\in H$ which covers at least $1/h$ relevant vertex pairs.
By the greedy choice of $w$, $w$ hits at least the same number of relevant pairs as $u$ does.

After $h$ consecutive vertices from $B_{2\cdot 2^i}(v)$ are selected, the
number of relevant vertex pairs is at most $(1-1/h)^h \le 1/e$ fraction of the original,
i.e., is reduced by a factor of $e$.
The initial number relevant vertex pairs is bounded by $n^2$,
therefore the algorithm chooses $O(h \log n)$ vertices in $B_{2\cdot 2^i}(v)$
before all relevant vertex pairs are hit.
Once all the relevant vertex pairs are hit, the algorithm will not choose
any level $i$ vertices in $B_{2\cdot 2^i}$.
\end{proofatend}

\section{Upper Bounds}
\label{sec:up}

In Sections~\ref{sec:up} and \ref{sec:lb}, we assume that isolated vertices are
deleted from the center graphs, so their density is the number of
edges divided by the number of (non-isolated) vertices.

\subsection{Greedy}\label{sec:ugreedy}

We show that g-HHL finds an \hhl\ of size that is within an $O(\sqrt n \log n)$
factor of the optimal \hl\ size.
We prove this by bounding  the ratio of the density of the center graph picked by g-HHL and the
density of the MDS of a center graph.

\begin{theorem}\label{thm:greedy}
g-HHL is an $O(\sqrt n \log n)$-approximation algorithm for \hl.
\end{theorem}
\begin{proof}
Suppose that at some iteration, the algorithm picks a center graph
with $m'$ arcs and $n'$ vertices. Then by the definition of g-HHL
all center graphs have at most $m'$ arcs, so the density of the maximum density
subgraph (over all center graphs) is at most $\sqrt{m'}$.
The density ratio of the maximum density subgraph to that of the chosen
center graph is at most
$$
\frac{\sqrt{m'}}{m'/n'} = \frac{n'}{\sqrt{m'}} \le
\frac{n'}{\sqrt{n'/2}} = \sqrt{2n'} \le \sqrt{2n} .
$$
Here we use the fact that the chosen graph has no isolated vertices,
so $m' \ge n'/2$.
It follows that the density of the chosen center graph is a $\sqrt{2n}$-approximation of the
maximum density of any subgraph.
By Lemma~\ref{lm:appr} we have that the labeling size is larger than the size of the optimal \hl\ by at most $O(\sqrt n \log n)$ factor.
\end{proof}

Since \hhl\ is a special case of \hl\ we have
\begin{corollary}
g-HHL is an $O(\sqrt n \log n)$-approximation algorithm for \hhl.
\end{corollary}

\subsection{Distance Greedy} \label{sec:udgreedy}

We show that d-HHL finds an \hhl\ of size within an $O(\sqrt{n}\log n\log D)$
factor of the optimal \hl\ size. But first we need to extend our concept of hub labels.

Cohen et al.~\cite{CHKZ-03} defined a more general notion of {\em hub labels for a given set $U$ of vertex pairs}.
Such labels are required to have a vertex $w \in L(u)\cap L(v)$ which is on a shortest path between $u$ and $v$ for each
$[u,v]\in U$. The $O(\log n)$ approximation algorithm described in Section \ref{sec:approx}
works for this more general notion of \hl; Lemma~\ref{lm:appr} and
Theorem~\ref{thm:greedy} hold.

\begin{theorem}
d-HHL is an $O(\sqrt n \log n\log D)$-approximation algorithm for \hl.
\end{theorem}
\begin{proof}
Let \opt\ denote the size of the  optimal \hl.
Let $U_i$ be a set of vertex pairs at level~$i$ which are not covered by vertices at higher levels when we run d-HHL.
Let $\hl_i$ be the optimal \hl\ to cover vertex pairs from $U_i$ and let $\opt_i$ be size of $\hl_i$.
Since $U_i$ is a subset of all vertex pairs, $\opt_i$ doesn't exceed \opt.
By Theorem~\ref{thm:greedy} we can use the g-HHL algorithm to find $O(\sqrt n \log n)$ approximation for $\hl_i$.

Now let's return to d-HHL.
Since every two pairs at the same level have the same weight and weights of all lower-level vertex pairs are negligible,
at the consecutive set of iterations in which d-HHL
covers $U_i$ it picks the same vertices as  g-HHL when we run it on $U_i$.

So the labels found by d-HHL have size
$$
\sum_{i=0}^{\lfloor\log D\rfloor} O(\sqrt n\log n)\opt_i \le \sum_{i=0}^{\lfloor\log D\rfloor} O(\sqrt n\log n)\opt = O(\sqrt n\log n\log D)\opt.
$$
\end{proof}

\begin{corollary}
d-HHL is an $O(\sqrt n \log n\log D)$-approximation algorithm for \hhl.
\end{corollary}

\subsection{Weighted Greedy} \label{sec:uwgreedy}

Although w-HHL is motivated by the approximation algorithm
of Cohen et al., it does not achieve $O(\log n)$ approximation.
We show that w-HHL finds an \hhl\ of size larger than the size of the optimal \hl\ by  an $O(\sqrt n \log n)$
factor.
The key to the analysis is the following lemma.
\begin{lemma} \label{lm:density}
If $G(V,E)$ is a graph with no isolated vertices, then
$G$ is an $O(\sqrt{n})$-approximation of the maximum density
subgraph of $G$.
\end{lemma}
\begin{proof}
Consider a subgraph $(V', E')$ of $G$.
Let $|V| = n$, $|E| = m$, $|V'| = n'$, $|E'| = m'$.
Then
$$
m' \le \min(m, {n'}^2) = n' \min \left(\frac{m}{n'} , n' \right) \le n' \sqrt m .
$$
where the last step follows since if $n' \le \sqrt m$,
$\min \left(\frac{m}{n'} , n' \right) = n' \le\sqrt m$, and if
$n' > \sqrt m$, $\min \left(\frac{m}{n'} , n' \right) = \frac{m}{n'} \le\sqrt m$.

Since $G$ goes not have isolated vertices, $m \ge n/2$, so we have
$$
\frac{m'}{n'} \le \sqrt m = \frac{m}{n} \frac{n}{\sqrt m} \le
\frac{m}{n} \frac{n}{\sqrt{n/2}} \le
\frac{m}{n} \sqrt{2n} .
$$
\end{proof}

\begin{theorem}
w-HHL is an $O(\sqrt n \log n)$-approximation algorithm for \hl.
\end{theorem}
\begin{proof}
At each iteration, w-HHL picks the center graph
with the maximum ratio of the number of edges divided by the number of
vertices. By Lemma \ref{lm:density} the density of this graph is smaller than the density of the densest subgraph of
a center graph by at most $O(\sqrt n)$.
Therefore  by Lemma \ref{lm:appr} w-HHL produces an \hhl\ of size within an  $O(\sqrt n \log n)$
factor of the size of the  optimal \hl.
\end{proof}

\begin{corollary}
w-HHL is an $O(\sqrt n \log n)$-approximation algorithm for \hhl.
\end{corollary}

\section{Lower Bounds}
\label{sec:lb}

In this section we show that g-HHL, d-HHL and w-HHL do not give a poly-log approximation.
We present graphs for which these algorithms find a labeling worse than the optimal \hhl\ by a polynomial factor.
We also show that our upper bounds are fairly tight.

\subsection{Greedy}\label{sec:lgreedy}

We show that for a graph in Figure~\ref{fig:bad-g} g-HHL finds a labeling larger by
 a factor of  $\Omega(\sqrt n)$ than the optimal \hhl\ (and therefore the optimal \hl).

\begin{lemma}\label{tm:bad-g}
There is a graph family for which g-HHL finds \hhl\ of size $\Omega(n^{3/2})$ while the optimal \hhl\ size is $O(n)$.
\end{lemma}
\begin{proofatend}
Consider the directed graph $G=(V,A)$ in Figure~\ref{fig:bad-g}.
The graph $G$  has $n=\Theta(k^2)$ vertices $V=\{a_1,\ldots, a_k, b_1,\ldots, b_{k+1}\}\cup \{c_{ij} \mid 1\le i\le k+1, 1\le j\le k\}$. The arcs are $A=\{(a_i,b_j)\mid 1\le i\le k, 1\le j\le k+1\}\cup\{(b_i,c_{ij})\mid 1\le i\le k+1, 1\le j\le k\}$ all of length $1$.

Consider the center graphs when g-HHL starts and the labeling is empty.
Shortest paths containing $a_i$ include the path from $a_i$ to itself, the
paths from $a_i$ to $b_x$, and the paths from $a_i$ to $c_{xy}$, so
number of edges in the center graph of $a_i$ is
$$
1 + (k+1) + k(k+1) =(k+1)^2 + 1 \ .
$$
Shortest paths containing $b_i$ include the path from $b_i$ to itself, $k$ paths from $a_j$ to $b_i$, another $k$ paths from $b_i$ to $c_{ij}$, and the $k^2$ paths from $a_j$ to $c_{ij}$ for a total of
$$
1 + k + k + k^2 = (k+1)^2 \ .
$$
Shortest paths containing $c_{ij}$ include the path from
$c_{ij}$ to itself, from $c_{ij}$ to $b_i$, and the $k$ paths from
$c_{ij}$ to $a_x$ for a total of
$$
1 + 1 + k = k+2 \ .
$$
So g-HHL will pick an $a_i$ vertex for some $i$ first.
Note that if when g-HHL  picks an $a_i$ vertex, the center graph of
$a_j$, $j\not= i$ does not change, and the center graphs
of the $b_i$'s and the $c_{ij}$'s  loose edges.
Therefore g-HHL will continue picking $a$-vertices until there are none left.
After that, a center graph of some $b_i$ has $k+1$ edges and a center graphs of $c_{ij}$
has $2$ edges.
So g-HHL will pick all $b$ vertices next, and then all the
$c$ vertices. 

The order found by g-HHL is $a_1,\ldots,a_k,b_1\ldots,b_{k+1}$ followed by $c$ vertices and
the labeling it produced is as follows.
$|L_f(a_i)| = |L_b(a_i)| = 1$, $|L_f(b_i)| = 1 + k$, $L_b(b_i) = 1 + k$,
$|L_f(c_{ij})| = 1$, and $|L_b(c_{ij})| = 2 + k$.
Therefore the total size of the labeling is $\Omega(k^3)$ = $\Omega(n^{3/2})$.

A better order for this graph is $b_1,\ldots,b_{k+1},a_1,\ldots,a_k$ followed by $c$ vertices.
The canonical labeling corresponding to this order is as follows.
 $|L_f(a_i)| = (k+1) + 1$, $|L_b(a_i)| = 1$,
$|L_f(b_i)| = |L_b(b_i)| = 1$, $|L_f(c_{ij})| = 1$, and
$|L_b(c_{ij})| = 2$.
The total size of this labeling is $O(k^2) = O(n)$.
\end{proofatend}

\begin{figure}
\begin{center}
    \subfloat[Bad example for g-HHL and d-HHL.\label{fig:bad-g}]{\includegraphics[scale=1,page=1,height=6em]{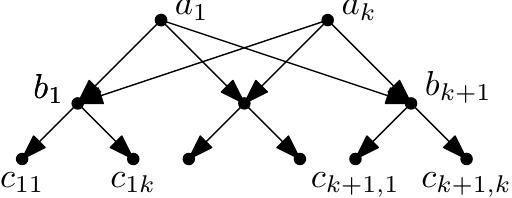}}
\qquad
\qquad
\qquad
    \subfloat[Bad example for w-HHL.\label{fig:bad-w}]{\includegraphics[scale=1,page=2,height=6em]{figures.pdf}}
\end{center}
\caption{Bad examples for greedy \hhl\ algorihtms.}\label{fig:bad}
\end{figure}

We have shown that for $G=(V,E)$, g-HHL
produces a labeling larger than the optimal one by
an  $\Omega(\sqrt n)$ factor, so
our $O(\sqrt n \log n)$ upper bound on the approximation ratio of g-HHL of Section \ref{sec:ugreedy} is
tight up to a logarithmic factor.

\subsection{Distance Greedy}

We show that for a graph in Figure~\ref{fig:bad-g} d-HHL finds a labeling larger by
 a factor of  $\Omega(\sqrt n)$ than the optimal \hhl\ (and therefore the optimal \hl).

\begin{lemma}
There is a graph family for which d-HHL finds \hhl\ of size $\Omega(n^{3/2})$ while the optimal \hhl\ size is $O(n)$.
\end{lemma}
\begin{proofatend}
Consider the directed graph $G=(V,A)$ depicted in Figure~\ref{fig:bad-g}.
There are paths of length 0, 1, and 2.
While there are some paths of length 2 yet uncovered, d-HHL selects a vertex to hit
the maximum number of paths with length 2.
Weights of all paths of length 0 and 1 matter only when d-HHL chooses between two vertices which hit exactly the same number of paths of length 2.

At the beginning $a_i$ hits $k(k + 1)$ paths of length 2 and $b_i$ hits $k^2$ paths of length 2.
So d-HHL selects $a_i$. As d-HHL proceeds, the number of paths of length 2 hit by $b_i$ decreases
and the number of paths of length 2 hit by $a_i$ remains the same $k(k + 1)$.


Therefore the order found by d-HHL is $a_1,\ldots,a_k,b_1\ldots,b_{k+1}$ followed by all the $c$ vertices.
Exactly the same order is produced by g-HHL.
From Lemma~\ref{tm:bad-g} we know that
the size of the canonical labeling of this order is $\Omega(n^{3/2})$  while the size of optimal \hhl\ is $O(n)$.
\end{proofatend}

So d-HHL can also produce a labeling of size $\Omega(\sqrt{n})$ away from optimal.
This makes a fairly good match with the $O(\sqrt n\log n\log D)$ upper bound established in Section \ref{sec:udgreedy}.

Theorem~\ref{tm:dhhl} gives $O(h \log n\log D)$ bound for the maximum label size produced by d-HHL.
The graph in Figure~\ref{fig:bad-g} gives us a good lower bound on the maximum label size as the following lemma specifies.

\begin{lemma}\label{lm:bad-h}
There is a graph family for which d-HHL finds \hhl\ with maximum label size $\Omega(h \log D)$.
\end{lemma}
\begin{proofatend}
Consider the directed graph $G=(V,A)$ in Figure~\ref{fig:bad-g}. The diameter of $G$ is 2. Let's find the highway dimension $h$ of $G$.

Abraham et al.~\cite[Lemma 3.5]{ADFGW-13} show that the maximum degree of a vertex is a lower bound
on the HD.
Thus $h$ is at least $k+1$.
Note that all $b_i$ form a hitting set of size $k+1$ for all paths of length greater than 0.
So any $S_r(v)$ has a hitting set with at most $k+2$ vertices and thus $h = \Theta(k)$.

In the labels found by d-HHL we have $|L(c_i)| = k+2$ (cf. the proof of Lemma~\ref{tm:bad-g}).
Since $h= \Theta(k)$ and $D=2$ we have $|L(c_i)| = \Theta(h \log D)$.
\end{proofatend}
Lemma~\ref{lm:bad-h} shows that the upper bound of Theorem~\ref{tm:dhhl} is right up to a $O(\log n)$ factor.

\subsection{Weighted Greedy}

We show that for a graph in Figure~\ref{fig:bad-w} w-HHL finds a labeling of size larger than the size of the optimal \hhl\  by a
factor of $\Omega(\sqrt[3] n)$ (and therefore the optimal \hl).

\begin{lemma}
There is a graph family for which w-HHL finds \hhl\ of size $\Omega(n^{4/3})$ while the optimal \hhl\ size is $O(n)$.
\end{lemma}
\begin{proofatend}
Consider the undirected graph $G=(V,E)$ in Figure~\ref{fig:bad-w} .
The  vertices of $G$ are
$V=\{a, b\}\cup
\{c_{i} \mid 1\le i\le k\} \cup
\{d_{ij} \mid 1\le i\le k, 1\le j\le l\}$, so $n=|V|=\Theta(kl)$.
The edges are
$
E=\{(a,d_{ij})\mid 1\le i\le k, 1\le j\le l \} \cup
\{(b,c_{i})\mid 1\le i\le k \} \cup
\{(c_i,d_{ij}\mid 1\le i\le k, 1\le j\le l)\}
$
All edges have length $2$ except for those adjacent to $a$,
which have length $3$.
The lengths of the edges are set so that shortest paths between distinct $d$ vertices
adjacent to the same $c$ vertex go through the $c$ vertex.

We set $l = 2k^2$; so $k = \Theta(\sqrt[3]n)$.
As we shall see, this is large enough to make w-HHL
choose the $c$ vertices before choosing $b$.
However, this causes the  $c$ vertices to be added to the labels of many $d$
vertices and leads to a large total label size.

Consider the center graphs when w-HHL starts and the labeling is empty.
Since the graph is connected, all center graphs have no isolated vertices,
so all the denominators of  the densities of the center graphs  are the same and equal $n$.

Consider now the numerators (number of pair covered) by the different vertices. Vertex
$a$ covers the shortest paths between the $d$ vertices
adjacent to different $c$'s. Therefore
the center graph of $a$ has $\Omega((kl)^2)$ edges, which is asymptotically more
than the number of edges in the other center graphs. So w-HHL chooses $a$
to be the most important vertex.

Following this first choice of $a$, all vertex pairs consisting of $a$ and $d$'s are covered,
except for the pairs of $d$'s of the form $d_{ij}$ and $d_{ir}$ (both adjacent to $c_i$).
The vertex $d_{ij}$ is an endpoint of
every uncovered shortest path containing
it
 and therefore the density of the center graph of $d_{ij}$
is constant.
As we shall see, the density of other center graphs is higher, so the $d$ vertices
are chosen last.

We show that after choosing $a$, w-HHL chooses $c$ vertices until
there are no $c$ vertices left.
Suppose the number of remaining $c$ vertices is $t: 1 \le t \le k$.
We  show that the density of the center graph of each of the remaining $c$'s is larger than the density of the center graph of $b$. 
First we observe that at this point the number of vertices in the center graph of $b$ and in the center graph of each of the remaining $c$ vertices is the same, namely $1 + t + tl$.

Shortest paths through $b$ include the paths between
$c_i$ and $c_j$ for $i < j$, the shortest paths from $c_i$ to $d_{jr}$
for $i \ne j$ and paths from $b$ to all the vertices that have not been picked yet.
So the number of edges in the center graph of $b$ is
\begin{equation} \label{eq:b-vertex} t(t-1)/2 + t(t-1)l + (1 + t + tl)\ . \end{equation}
Shortest paths through $c_i$ include the paths between $d_{ij}$
and $d_{ir}$ for $j < r$,
the paths from $d_{ir}$ to $c_j$ for $j\ne i$,
the paths from $b$ to $d_{ij}$,
and the paths from $c_i$ to all the vertices that have not been picked yet.
So the number of edges in the center graph of a remaining  $c$ vertex is
\begin{equation} \label{eq:c-vertex} l(l-1)/2 + l(t-1) + l + (1 + t + tl) \ . \end{equation}
Subtracting Equation (\ref{eq:b-vertex}) from Equation (\ref{eq:c-vertex}), and using the facts that
$l = 2k^2$ and $k \ge t$, we get

\begin{align*}
& l(l-1)/2 + lt - t(t-1)/2 - t(t-1)l = \\
& l^2/2 + 2lt + t/2 - l/2 - t^2/2 - t^2l \ge \\
& l^2/2 - t^2l = 2k^4 - 2k^3 > 0.
\end{align*}

So w-HHL chooses $a$ first, followed by all $c$ vertices, then  $b$  and all $d$ vertices.
The size of the corresponding canonical labeling is
$$
n + \sum_{t=1}^k (1 + t + tl) + 1 + kl = \Omega(lk^2) =
\Omega(n^{4/3}) .
$$
A better ordering is the one which puts $a$ is first, followed by $b$, the
$c$ vertices, and finally the $d$ vertices.
The size of the corresponding canonical labeling is
$$
n + (1 + k + kl) + k(1 + l) + kl = O(n) .
$$
\end{proofatend}

Therefore on the graph in Figure~\ref{fig:bad-w} the size of the labeling produced by w-HHL is larger than the optimal by  a factor of $\Omega(\sqrt[3] n)$.
There is a factor of $\Omega(\sqrt[6] n \log n)$ gap between this lower bound
and the $O(\sqrt n \log n)$ upper bound of Section \ref{sec:uwgreedy}.

\section{NP-Completeness}
\label{sec:np}

\subsection{Undirected Graphs}

In this section we prove that the problems of finding an optimal \hl\ and an optimal \hhl\ are NP-hard by a reduction from Vertex Cover (VC). The reduction takes an instance of VC
consisting of a graph $G$ and an integer $k$ and produces an undirected graph $G'$ and an integer $k'$ such
that the following conditions are equivalent
\begin{enumerate}
  \item There is an \hl\ of size $k'$ in $G'$.
  \item There is an \hhl\ of size $k'$ in $G'$.
  \item There is a VC of size $k$ in $G$.
\end{enumerate}
Our results imply NP-completeness of \hl\ and \hhl\ in
undirected graphs.

Before presenting the reduction we prove the following useful lemma.
\begin{lemma}\label{lm:star}
Let $G=(V,E)$ be a graph and $S$ be a star graph, distinct from $G$,
with a root $s$ and $|V|$ leaves.
Let $G'$ be the union of the graphs $G$ and $S$, with additional edges between
$s$ and some vertices of $V$. If $G'$ is connected then
 there are optimal \hl\ and \hhl\ for $G'$ such that
$s \in L(x)$ for every vertex $x$.
\end{lemma}
\begin{proofatend}
Let $L$ be an optimal \hl\ (or \hhl) labeling of $G'$.
First, assume that for a leaf $u \in S$ we have that $s \not \in L(u)$.
Since $(s,u)\in G'$ we must have that $u \in L(s)$, and the only pair of vertices covered by $u \in L(s)$
is $[s,u]$.
So if we add $s$ to $L(u)$ and remove $u$ from $L(s)$, we get a valid labeling
of the same size as $L$ which is  optimal.
Therefore we may assume that  $s \in L(u)$ and
$u \not \in L(s)$ for every leaf $u\in S$.

Next, assume that for some $v \in V$, and a leaf $u\in S$ we have that
$u \in L(v)$.
Since $u\in L(v)$ is used only to cover the pair $[u,v]$, we can
remove $u$ from $L(v)$ and add $s$ to $L(v)$ if it is not already there while keeping the labeling valid.
This way we can transform $L$, without increasing its size, to a labeling
such that the labels of $v \in V$ do not contain leaves of $S$.

Finally,
assume  that  there is a vertex $v\in V$ such that $s\not\in L(v)$.
Then $L(v)\cap S =\emptyset$. Since the pair $[u,v]$
for every leaf $u\in S $ has to be covered, $L(u)$ must contain a vertex of $V$.
Remove vertices of $V$ from $L(u)$ for all $u\in S$ and add $s$ to $L(v)$
for all vertices $v \in X$ such that $s$ is not in $L(v)$ already. This keeps the labeling valid
and cannot increase its size.
\end{proofatend}

Now we describe the reduction.
We reduce the problem of deciding whether there is a VC of size at most $k$ in $G$ to the problem of
deciding whether there is an \hl\ of size at most $k'$ in a graph $G'$. Lemma \ref{lm:npc-hhl} shows that
$G'$ has an \hl\ of size at most $k'$ iff it has an \hhl\ of size at most $k'$, so it follows that our reduction also proves that deciding whether there is an \hhl\ of a given size is also NP-complete.
We construct $G'=(V',E')$ from $G=(V,E)$ as follows.
\begin{enumerate}
    \item For each vertex $v\in V$ we add  three vertices, $v_1$, $v_2$, and $v_3$ to $V'$ and two edges $\{v_1,v_2\}$ and $\{v_2,v_3\}$ to $E'$.
    \item For each edge $\{u,v\}\in E$ we add an edge $\{u_1,v_1\}$ to $E'$.
    \item We add a star $S$ with $3|V|$ leaves and a root $s$ to $G'$ and
add $\{s,v_1\}$ to $E'$ for every $v\in V$.
\end{enumerate}

\begin{figure}
  \captionsetup[subfloat]{singlelinecheck=false}
  \begin{center}
    \begin{minipage}{.35\textwidth}
    \subfloat[Graph $G'$\label{fig:npcg}]{\includegraphics[page=4]{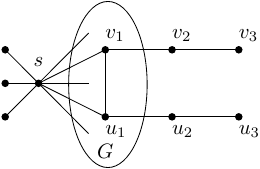}}

    \begingroup
    \captionsetup[subfloat]{width=\textwidth}
    \subfloat[$G'_u$ labels for a type~1 vertex\label{fig:npcv1}]{\includegraphics[page=9,scale=1.1,width=3cm]{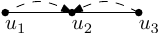}}
    \\
    \subfloat[$G'_v$ labels for a type 2 vertex\label{fig:npcv2}]{\includegraphics[page=10,scale=1.1]{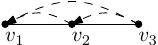}}
    \endgroup
    \end{minipage}
    \begin{minipage}{.5\textwidth}
    \subfloat[Subgraph $G'_{uv}$\label{fig:npcp}]{\includegraphics[page=5,scale=1.1]{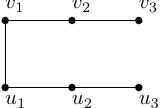}}
    \qquad
    \qquad
    \subfloat[$G'_{uv}$ labels when $u$ is type~1 and $v$ is type~2\label{fig:npcpl1}]{\includegraphics[page=7,scale=1.1]{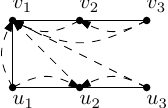}}
    
    \subfloat[$G'_{uv}$ partial labels when both $v$ and $u$ are type~1\label{fig:npcpl2}]{\includegraphics[page=6,scale=1.1]{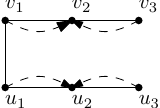}}
    \qquad
    \qquad
    \subfloat[$G'_{uv}$ labels when both $v$ and $u$ are type~2\label{fig:npcpl3}]{\includegraphics[page=8,scale=1.1]{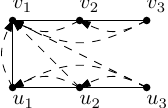}}

    \end{minipage}
  \end{center}
  \caption{Reduction from VC to HL}
\end{figure}

The graph $G'$ is shown in Figure~\ref{fig:npcg}.
All edges have length 1.
By Lemma~\ref{lm:star} we can assume w.l.g.\ that in an optimal labeling all vertices have $s$ in their labels.
Therefore, all paths between $v_i$ and $u_j$ such that $\{u,v\} \notin E$ are covered (hereinafter when we write $v_i$ we mean $v_i$ for $i=1,2,3$).

For each vertex $v\in V$ we have a subgraph $G'_v$ in $G'$ which is a path $(v_1,v_2,v_3)$.
Any labeling must cover all $[v_i,v_j]$ pairs of $G'_v$.
We show that w.l.g the labeling covers these paths either as in Figure~\ref{fig:npcv1}
in which case we say that $v$ is a \emph{type 1} vertex or as in
Figure~\ref{fig:npcv2} in which case we say that $v$ is a \emph{type 2} vertex.
Note that a type 2 vertex uses one more hub in the labeling, so to reduce the labeling size we want to minimize the number of
type 2 vertices. We will show, however, that the type 2 vertices must form a vertex cover for the labeling to be valid.

\begin{lemma}\label{lm:type}
There is an optimal labeling $L$ of $G'$ such that for each vertex $v\in V$ if $v_1\in L(v_2)$, then $v$ is a type 2 vertex, otherwise $v$ is a type 1 vertex.
\end{lemma}
\begin{proofatend}
Vertex $v_3\in L(v_2)$
can cover only the pair $[v_2,v_3]$. So
if $v_3\in L(v_2)$ we can remove $v_3$ from $L(v_2)$ and put $v_2$ in $L(v_3)$ instead.
Similarly
$v_3\in L(v_1)$
can cover only the pair $[v_1,v_3]$. So if $v_3\in L(v_1)$ we can remove $v_3$ from $L(v_1)$ and put $v_1$ in $L(v_3)$ instead.
Now if $v_1\in L(v_3)$ then we don't need $v_2\in L(v_1)$ and can replace $v_2\in L(v_1)$ by $v_1\in L(v_2)$
keeping $L$ optimal and making
 $v$  a type 2 vertex.

If $v_1\notin L(v_3)$ then we have $v_2\in L(v_1)$ to cover the pair $[v_1,v_3]$. So either $v$ is a type 1 or
 there an additional hub $v_1\in L(v_2)$. In the latter case we can remove $v_2$ from $L(v_1)$ and put $v_1$ into $L(v_3)$, making $v$ a type 2 vertex.
\end{proofatend}

For an edge $\{u,v\}\in E$ let $G'_{uv}$  be the subgraph of  $G'$ corresponding to this edge as shown in Figure~\ref{fig:npcp}. $G'_{uv}$ contains all shortest paths between $v_i$ and $u_j$.
Note that no vertex of $G'$  other than $u_i$ and $v_j$ hits these paths.
We say that a hub $u_i \in L(v_j)$ or $v_i \in L(u_j)$ is a \emph{$\{u,v\}$-crossing} hub.

\begin{lemma}\label{lm:edge}
If there is an edge $\{u,v\}\in E$ then the labels of $u_i, v_i$, $1\le i\le 3$ contain at least 3 $\{u,v\}$-crossings.
\end{lemma}
\begin{proofatend}
Consider three pairs: $[u_1,v_1], [u_2,v_2]$ and $[u_3,v_3]$.
To cover each $[u_i,v_i]$ pair for $i=1,2,3$ we need a $\{u,v\}$-crossing hub.
So $L(u_i)\cup L(v_i)$ contains a $\{u,v\}$-crossing hub.
Since all three $L(u_i)\cup L(v_i)$ are disjoint, $L$ has at least 3 $\{u,v\}$-crossing hubs.
\end{proofatend}

The following lemma shows that the type 2 vertices must form a VC.

\begin{lemma}\label{lm:path}
There is an optimal labeling $L$ for $G'$ such that for each edge $\{u,v\}\in E$ there is at least one type 2 vertex among $u$ and $v$.
\end{lemma}
\begin{proofatend}
By Lemma~\ref{lm:type} we can assume that every vertex is either a type 1 or  a type 2 vertex in $L$.

Suppose $\{u,v\}\in E$ and  both $u$ and $v$ are type 1 vertices.
A partial labeling is shown in Figure~\ref{fig:npcpl2}.
Since $u_1\notin L(u_2)$, $u_1$ cannot cover the pair $[v_1,u_2]$.
Similarly, $v_1$ cannot cover the pair $[u_1,v_2]$  and neither $u_1$ nor $v_1$ can cover
the pair $[u_2,v_2]$.
With one more hub to cover the pair $[u_1,v_1]$ it follows that we need at least 4 different $\{u,v\}$-crossing hubs already.
Let $u_2\in L(v_2)$ (the case $v_2\in L(u_2)$ is similar).
Then we need one more $\{u,v\}$-crossing hub to cover the pair $[u_2,v_3]$ and the total number of hubs
to cover shortest paths in $G'_{uv}$ is at least~9.

If we make $v$ a type 2 vertex then 8 hubs suffice as shown in Figure~\ref{fig:npcpl1}.
We didn't change other hubs in labels of $u_i$ so all $[u_i,w_j]$ pairs for $w\ne v$ remain covered.
Also all $[v_i,w_j]$ pairs remain covered, since $v_2$ can't be used as a hub for any $[v_1,w_j]$ pair.
\end{proofatend}

The following lemma gives a reduction from VC to \hl.
\begin{lemma}\label{lm:npc}
The graph $G$ has a VC of size $k$ if and only if $G'$ has an \hl\ of size $12|V| + 1 +3|E| +k$.
\end{lemma}
\begin{proofatend}
Assume $G$ has a vertex cover of size at most $k$.
We construct an \hl\ of $G'$ as follows.
We put $s$ and $v$ itself in $L(v)$ for every $v\in V'$.
Since there are $6|V| +1$ vertices in $G'$ this contributes $12|V|+1$ hubs.
Then we make each vertex of the
  vertex cover a
 type 2 vertex and each vertex which is not in the vertex cover  a type 1 vertex.
We use 2 hubs to cover $G'_v$ for a type 1 vertex and 3 hubs for a type 2 vertex, for the total of $2|V|+k$ hubs.
For each edge $\{u,v\}\in E$ we use 3 $\{u,v\}$-crossing hubs to cover $G'_{uv}$ as shown in Figure~\ref{fig:npcpl1} and Figure~\ref{fig:npcpl3}.
So the total labeling size is $12|V| + 1 +3|E| +k$.

Now assume that $L$ is an optimal \hl\ of $G'$ of size $12|V| + 1 +3|E| +k$.
By Lemma~\ref{lm:star}
we know that any vertex $w\in G'$ has $s$ in its label and by
 Lemma~\ref{lm:type} we know that there exists such an $L$ that makes every vertex $v\in V$  either a type 1 or a type 2 vertex.
By Lemma~\ref{lm:edge} we know that there are at least 3 $\{u,v\}$-crossing hubs for any edge $\{u,v\}\in E$.
Since the size of $L$ is at most $12|V| + 1 +3|E| +k$ it follows that there are at most $k$ type 2 vertices in $L$.
Lemma~\ref{lm:path} implies that these $k$ vertices form a vertex cover.
\end{proofatend}

\begin{theorem} \label{tm:npc}
The problem of deciding whether an undirected graph has an \hl\ of size at most $k$ is NP-complete.
\end{theorem}

The following lemma shows that our reduction is in fact also a valid reduction from  VC to finding an optimal \hhl.

\begin{lemma}\label{lm:npc-hhl}
The graph $G'$ has an \hl\ of size $k'$ if and only if it has an \hhl\ of size $k'$.
\end{lemma}
\begin{proofatend}
The ``if'' part follows from the fact that every \hhl\ is an \hl.
For the ``only if'' part consider an optimal \hl\ $L$ of size at most $k'$.
By Lemma \ref{lm:type} each vertex is either of type 1 or of type 2.
Consider the following order of the vertices of
$G'$. The most important vertex is $s$ followed by all the leaves of $S$.
Then we put the triple $v_1,v_2,v_3$ for all type 2 vertices where for each $v$, $v_1$ is more important than $v_2$ which is more important than $v_3$ and the order of the triples corresponding to different vertices is arbitrary.
Finally put a triple $v_2,v_1,v_3$ for all type 1 vertices where for each $v$, $v_2$ is more important than $v_1$ which is more important than $v_3$ and the order of the triples corresponding to different vertices is arbitrary.
The labels in Figure~\ref{fig:npcpl1} and Figure~\ref{fig:npcpl3} respect this order.
Thus the \hhl\ $\tilde{L}$ corresponding to this order has exactly~3 $\{u,v\}$-crossings for each $\{u,v\}\in E$.
Therefore by Lemma~\ref{lm:edge} $\tilde{L}$ has the same size as $L$.
\end{proofatend}

Lemma~\ref{lm:npc-hhl} and Lemma~\ref{lm:npc} immediately imply the following
\begin{theorem} \label{tm:npc-hhl}
The problem of deciding whether an undirected graph has an \hhl\ of size at most $k$ is NP-complete.
\end{theorem}

Theorem~\ref{tm:npc} and Theorem~\ref{tm:npc-hhl} show that both \hl\ and \hhl\ are NP-Complete in undirected graphs with unit lengths.
If we change length of edges $\{s,v_1\}$ for $v\in V$ from $1$ to $0.9$ our proof is not affected. However, the shortest paths in $G'$ become unique.
So \hl\ and \hhl\ are NP-Complete in undirected graphs even when shortest paths are unique.

\subsection{Directed Graphs}

Here we show that both optimal \hl\ and \hhl\ are NP-hard in directed graphs.
We begin with \hhl, for which there is a simple reduction from the undirected case.

Let $G$ be an undirected graph. We transform $G$ to directed graph $G'$ by replacing each edge $\{u,v\}$ with two arcs $(u,v)$ and $(v,u)$. Now we present the reduction.
\begin{lemma}
The graph $G$ has an \hhl\ of size $k$ if and only if $G'$  has an \hhl\ of size $2k$.
\end{lemma}
\begin{proofatend}
To show the ``only if'' part, we take the labeling $\tilde L$ constructed from $L$ as follows $\tilde L_f(v) := L(v)$ and $\tilde L_b(v) := L(v)$.

Now we show the ``if'' part.
We can assume that $\tilde L_f, \tilde L_b$ is a canonical labeling (or replace the labeling
by a smaller canonical one).
Since $G'$ is symmetric, from the definition of canonical labeling it follows that for any vertex $v$ the forward label has exactly the same hubs as the backward label.
Moreover the distances from $v$ to and from the hubs are the same. So $L$ defined as $L(v) := \tilde L_f(v)$ is a valid labeling for $G$.
\end{proofatend}
\begin{theorem}
The problem of deciding whether a directed graph has an \hhl\ of size at most $k$ is NP-complete.
\end{theorem}

The following remark implies that the above reduction doesn't work for \hl.
\begin{remark}
For a directed graph a minimum \hl\ need not be symmetric.
\end{remark}
\begin{proof}
Consider the 4-cycle graph $C_4=(V,E)$, $V=\{v_0,v_1,v_2,v_3\}$, $E=\{\{v_i, v_{i+1 \mod 4}\}\mid 0\le i\le 3\}$
and the corresponding directed graph~$C_4'$.
An \hl\ $L$ of size 16 for $C_4'$ is shown in Figure~\ref{fig:cycle} (for example $L_f(v_0)$ contains $v_0$ and $v_3$ and $L_b(v_0)$ contains $v_0$ and $v_1$). Note that it is not symmetric as for example $L_f(v_0)\ne L_b(v_0)$.
Any labeling in $C_4'$ satisfying $L_f = L_b$ correspond to a labeling in $C_4$ of half the size.
So in order to show that there is no symmetric labeling of $C_4'$ of size 16 we show that there is no labeling of $C_4$ of size at most 8.
Indeed we need 4 hubs to cover the pairs $[v_i,v_i]$ and 4 hubs to cover the pairs $[v_i,v_{i+1 \mod 4}]$.
This already counts for 8 hubs. Therefore no $v_i$ is in $L(v_{i\pm 2 \mod 4})$.
To cover the $[v_0,v_2]$ pair we need $v_1$ (or $v_3$, the case is similar) to be in both $L(v_0)$ and $L(v_2)$ and therefore $L(v_1)$ contains only $v_1$.
But now we have no common hub for the $[v_1, v_3]$ and therefore it is uncovered. So there is no \hl\ of size $8$ for $C_4$.
\end{proof}

\begin{figure}
\begin{center}
\begin{minipage}[t]{.48\textwidth}
\begin{center}
    \includegraphics[scale=1.1,page=13,height=8em]{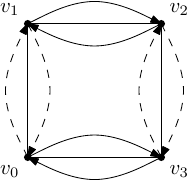}
\end{center}
\caption{Optimal \hl\ for directed graph $C_4'$. Solid and dashed arcs represent forward and backwards labels
respectively.
}\label{fig:cycle}
\end{minipage}
\quad
\begin{minipage}[t]{.48\textwidth}
\begin{center}
    \includegraphics[scale=1.1,page=12,height=8em]{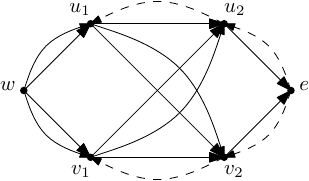}
\end{center}
\caption{Gadget corresponding to an edge $\{u,v\}$ in the reduction from VC to \hl\ in a directed graph (curly arcs represent labels).}\label{fig:hl}
\end{minipage}
\end{center}
\end{figure}

Now we present another reduction from VC to \hl\ in a directed graph.
For a VC instance $G=(V,E)$ we construct an HL instance $G'=(V',A')$,
$V'=\{w\} \cup \{v_1, v_2 \mid v \in V\} \cup \{e \mid e\in E\}$, $A'=\{(w,v_1), (v_1,v_2) \mid v \in V\}\cup \{ (u_1,v_2), (v_1,u_2), (u_2,e), (v_2, e) \mid e=\{u,v\}\in E\}$.
All arcs have length 1.
For each edge $e=\{u,v\}$ from $G$ we have a gadget as shown in Figure~\ref{fig:hl} (consider only straight arcs).

For any labeling we have $x$ in both $L_f(x)$ and $L_b(x)$ for all vertices $x$ and either $x\in L_b(y)$ or $y \in L_f(x)$ for all arcs $(x,y)$. Let us call such hubs \emph{mandatory} and all other hubs \emph{non-mandatory}.
Mandatory hubs cover all pairs  $[x,y]$ such that $\dist(x,y)\le 1$.
Any labeling for $G'$ has at least $M(G') = 2|V'| + |A'|$ mandatory hubs.

\begin{lemma}\label{lm:dir}
The graph $G$ has a VC of size $k$ if and only if $G'$  has an \hl\ of size $M(G') + k$.
\end{lemma}
\begin{proof}
We claim that mandatory hubs are enough to cover all pairs in $G'$ except $[w,e]$ for $e\in E$, which means all pairs $[x,y]$ with $\dist(x,y) \le 2$.
The sufficient labeling is shown in Figure~\ref{fig:hl} by curly arcs: a solid curly arc $(x,y)$ means $y\in L_f(x)$ and a dashed curly arc $(x,y)$ means $y\in L_b(x)$.
Indeed, for a pair $[x,y]$ with $\dist(x,y)= 2$ we have either $x=w$ or $y=e$ for some $e\in E$. In the former case $y=u_2$ for some $u\in V$ and the common hub is $u_1$. In the latter case $x=u_1$ for some $u\in V$ and either $e = \{u,v\}$ or $e = \{v,v'\}$ for some neighbor $v\in V$ of $u$. In both cases $v_2$ is the common hub.

Since $\dist(w,e)=3$ for a $e\in E$ we need a non-mandatory hub to cover a $[w,e]$ pair.
The non-mandatory hubs correspond to the vertex cover in $G$.
If there is a VC of size $k$ in $G$ then it is sufficient to use exactly $k$ non-mandatory hubs: add $v_2$ to $L_f(w)$ for every $v$ in VC.

Suppose there is an \hl\ with at most $k$ non-mandatory hubs.
We build a VC of size at most~$k$.
For a non-mandatory hub $e \in L_f(w)$ and any non-mandatory hub in $L_b(e)$ for an edge $e=\{u,v\}\in E$, add $u$ to the VC.
For a non-mandatory hub $v_2 \in L_f(w)$ for some $v\in V$, add $v$ to the VC.
It is easy to see that this is indeed the vertex cover.
\end{proof}
\begin{theorem}
The problem of deciding whether a directed graph has an \hl\ of size at most $k$ is NP-complete.
\end{theorem}

\section{HL vs.\ HHL}

\begin{figure}
\begin{center}
\begin{center}
    \includegraphics[scale=1.1,page=3,height=8em]{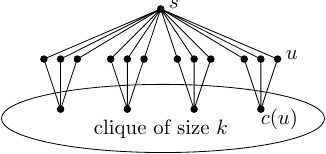}
\end{center}
\caption{Example which separates \hl\ and \hhl.}\label{fig:vs}
\end{center}
\end{figure}

In~\cite{GRS-13}, it is shown that the gap between the size of the  optimal
\hhl\ and the size of the  optimal \hl\ can be $\Omega(n^{0.26})$.
We show that for a graph in  Figure~\ref{fig:vs} the gap is  $\Omega(n^{0.5})$.

\Xcomment{
\begin{figure}
\begin{center}
    \includegraphics[scale=1,page=3]{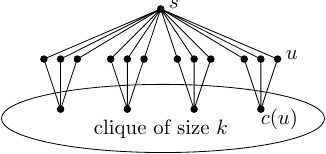}
\end{center}
\caption{Example which separates \hl\ and \hhl.}\label{fig:vs}
\end{figure}
}

\begin{theorem}\label{thm:hlvshhl}
There is a graph family for which the optimal \hhl\ size is $\Omega(\sqrt n)$ times larger than the optimal \hl\ size.
\end{theorem}
\begin{proofatend}
Consider the undirected graph shown in Figure~\ref{fig:vs}.
The graph consists of $k$ distinct stars each with $k-1$ leaves.
The centers of the stars are connected such that they form a clique.
Finally, there is an additional vertex $s$ connected to the leaves of all stars. The total number of vertices is $n = k^2+1$.
The length of every edge is $1$.

Consider the following \hl\ for this graph.
The vertex $s$ is in every label.
A center $v$ of a star $S$ is in the
 labels of all of the
vertices of $S$.
Finally, every star-center  has every other star-center in
its label.
It is easy to verify that the cover property holds for this labeling.
Each leaf $u$ of some star $S$ has a label of size $O(1)$.
The label of $s$ is of size $O(1)$. The size of the label of each star-center is
 $k+1$. It follows that the
 total size of this labeling is $O(n)$.

To construct an \hhl, we need to order the centers of the stars.
Fix such an order.
Consider a leaf $u$ of some star with center $c(u)$, and let $i$ be the number of star-centers
which are more important than $c(u)$.
For each star-center $v$ that is more important than $c(u)$,
$(u, c(u), v)$ is the shortest path between $u$ and $v$, so
either $v$ is in $L(u)$ or $u$ is in $L(v)$.
This accounts to $i$ hubs in the labels due to the pair $u,v$.
The total contribution of such hubs to the size of the labeling is
$$
(k-1) \sum_{i=1}^{k-1} i = k (k-1)^2 /2 = \Omega(n^{3/2}) .
$$

If follows that the total size of any hierarchical labeling is $\Omega(n^{3/2})$.
This yields an $\Omega(\sqrt n)$ gap between the optimal \hl\
and the optimal \hhl.
\end{proofatend}

The results of Section~\ref{sec:up} imply that the gap in Theorem~\ref{thm:hlvshhl} is within $O(\log n)$ factor of the best possible.

\section{Concluding Remarks}

Our lower bounds for greedy algorithms show that in contrast with \hl\ the greedy algorithm does not give a poly-log approximation for \hhl.
This motivates the question of whether a poly-log approximation algorithm for \hhl\ exists.
Our lower bound for w-HHL is $\Omega(\sqrt[6]n)$ factor away from the upper bound,
which leaves the open question to determine the polynomial factor for the w-HHL algorithm approximation guarantee. 

On many problem classes g-HHL and w-HHL find labelings of size close
to that found by the $O(\log n)$-approximation algorithm
for \hl~\cite{DGSW-14}.
It would be interesting to get a theoretical explanation of this
phenomena, for example by proving a better approximation ratio for
g-HHL or w-HHL on natural classes of graphs.

{\small
\bibliographystyle{plain}
\bibliography{labels,../bib/all}
}

\end{document}